\documentclass[pdftex,a4paper,12pt]{article}

\usepackage{amssymb,amsmath,amsfonts,nomencl,mathrsfs}
\usepackage{amsthm}  
\usepackage[arrow, matrix, curve]{xy}
\usepackage[latin1]{inputenc}
\newcommand{\IC}{\mathbb{C}}
\newcommand{\IR}{\mathbb{R}}
\newcommand{\IMM}{\mathscr{M}}

\newcommand{\ILL}{\mathscr{L}}

\newcommand{\IHH}{\mathscr{H}}

\newcommand{\IFF}{\mathscr{F}}

\newcommand{\IJJ}{\mathscr{J}}

\newcommand{\IN}{\mathbb{N}}

\newcommand*{\longhookrightarrow}%
               {\ensuremath{\lhook\joinrel\relbar\joinrel\rightarrow}}
\newcommand{\pa}{\slash\slash}

\newcommand{\Id}{{\rm d}}

\newcommand{\f}{\frac}
\newcommand{\nn}{\nonumber}

\newcommand{\sm}{\sim_b}

\theoremstyle{plain}            
\newtheorem{theorem}{theorem}[section]

\newtheorem{Corollary}[theorem]{Corollary}
\newtheorem{Theorem}[theorem]{Theorem}
\newtheorem{Proposition}[theorem]{Proposition}
\newtheorem{Propandef}[theorem]{Proposition and definition}

\theoremstyle{definition}       
\newtheorem{Definition}[theorem]{Definition}

\newtheorem{Remark}[theorem]{Remark}

\begin{document}

\begin{titlepage}


 \title{Generalized Schr\"odinger semigroups on infinite graphs}

   \author{
Batu G\"uneysu\footnote{Institut f\"ur Mathematik, Humboldt-Universit\"at zu Berlin, Germany. E-mail: gueneysu@math.hu-berlin.de}, Ognjen Milatovic\footnote{Department of Mathematics and Statistics, University of North Florida, Jacksonville, FL 32224, USA. E-mail: omilatov@unf.edu}, Fran\c{c}oise Truc\footnote{Institut Fourier, UMR 5582 du CNRS Universit\'e de Grenoble I BP 74, 38402 Saint-Martin d'H\`eres, France. E-Mail: francoise.truc@ujf-grenoble.fr}
}



\end{titlepage}

\maketitle

\begin{abstract}
With appropriate notions of Hermitian vector bundles and connections over weighted graphs which we allow to be locally infinite, we prove Feynman-Kac-type representations for the corresponding semigroups and derive several applications thereof.
\end{abstract}

\setcounter{page}{1}

\section{Introduction}

The past few years have witnessed the emergence of numerous works devoted to spectral properties of \emph{scalar} Schr\"odinger operators on discrete graphs. In particular, we would like to mention the papers~\cite{Keller-Lenz-09, Keller-Lenz-10}, which developed an abstract and elegant framework of Dirichlet forms on discrete sets and studied questions such as stochastic completeness, absence of essential spectrum, self-adjointness, and the description of generators of the corresponding Schr\"odinger semigroups. We would also like to mention the papers~\cite{Golenia-11,GKS-13,Keller-Lenz-09, Keller-Lenz-10, MilTu} and the references therein, where the authors treat the fundamental problem of defining self-adjoint extensions of scalar (magnetic) Schr\"odinger operators on discrete graphs.

A new aspect of this scalar theory has been treated in~\cite{GKS-13}: There, the authors have proved a Feynman-Kac-It\^o formula for the semigroups corresponding to a very general class of magnetic scalar Schr\"odinger operators on possibly locally infinite weighted discrete graphs. The essential observation here is that the very special path properties of the underlying right-continuous Markov process (it is in fact a jump process) suggests a natural corresponding notion of line integrals with respect to magnetic potentials. Indeed, it turns out that one simply has to include the exponential of this line integral as an extra factor in the classical Feynman-Kac formula (which deals with nonmagnetic scalar Schr\"odinger operators) in order to deal with magnetic scalar Schr\"odinger operators.

On the other hand, as in the manifold case, it is absolutely natural to consider covariant Schr\"odinger operators that act on \emph{sections in vector bundles} over discrete graphs. Indeed, in the last two decades the notion of vector bundles and connections on graphs has attracted quite a bit of attention. For example, the latter setting has been considered in~\cite{Kenyon-11}, where being motivated by the classical matrix-tree theorem, which relates the determinant of the combinatorial Laplacian on a graph to the number of spanning trees, the author generalizes this result to discrete covariant Laplacians that act on one- and two-dimensional vector bundles and interprets the determinant of the operators under consideration in terms of so-called cycle rooted spanning forests. Furthermore, let us point out that a variant of this setting also appears naturally in the theory of molecular bonds~\cite{FKC-S-92}, where it turns out that the vibration modes of a molecule are precisely the eigenvalues of a discrete covariant Schr\"odinger type operator which is invariant under a group of symmetries of the molecule. Finally, let us also mention that covariant Laplacians play an essential role in the analysis of large data sets~\cite{Si-Vu-12}: Namely, the concept of \lq\lq{}vector diffusion maps\rq\rq{}, which as the name suggests is a refinement of usual (= scalar) reduction methods such as diffusion maps, turns out to be a powerful mathematical framework for analyzing large high-dimensional data sets. In this framework, if the vertices $i$ and $j$ of a graph represent (say, two-dimensional) images, then the $2\times 2$ orthogonal matrix  $O_{ij}$, attached to the edge $ij$, rotationally aligns image $i$ with image $j$. Furthermore, the weights $b_{ij}$ provide some measure of affinity between the images $i$ and $j$ when they are optimally aligned.

Now as there exist Feynman-Kac type formulae for covariant Schr\"odinger operators that act on sections in vector bundles over manifolds \cite{G2,G-10}, one might ask whether there is a natural geometric extension of the scalar Feynman-Kac-It\^o formula from~\cite{GKS-13} to the corresponding class of covariant operators on discrete graphs.\emph{ Our main result, Theorem \ref{FK}, states that this is indeed possible:} There, we consider a Hermitian vector bundle $F$ over a weighted graph $(X,b,m)$, together with a unitary $b$-connection $\Phi$ on $F$ (see Section \ref{SS:v-bundle} for the precise definitions), which in spirit of the above particular examples, should assign to any two neighbors $x\sm y$ of $(X,b)$ a unitary map $\Phi_{x,y}:F_x\to F_y$. Additionally, a potential $V$  is naturally given by a family of self-adjoint maps $V(x):F_x\to F_x$, $x\in X$, and we can consider the covariant Schr\"odinger operator which is formally given by
$$
H_{\Phi,V}f(x)=\>\frac{1}{m(x)}\sum_{ y}b(x,y)\Big( f(x)-\Phi_{y,x} f(y)\Big) +V(x)f(x)
$$
in the Hilbert space of square integrable sections $f$ in $F$. To $\Phi$ (and the underlying right-continuous Markov process) we can associate a stochastic parallel transport, and to the pair $(\Phi,V)$ there canonically corresponds a path ordered exponential (see Section \ref{SS:prob-prelim}). In Theorem \ref{FK} we state that the expectation value at time $t$ of the product of these two processes essentially represents the generalized Schr\"odinger semigroup $\mathrm{e}^{-tH_{\Phi,V}}$. Here, we can in fact allow locally infinite weighted graphs as well as a very general class of potentials, which should cover all applications. It should be noted that, when compared with the scalar situation from~\cite{GKS-13}, the core of the proof of Theorem \ref{FK} becomes much more technical as the underlying processes do not commute anymore, so that the required equalities (like semigroup properties etc.) are more subtle. Let us also point out that in the scalar case, our formula precisely boils down to the Feynman-Kac-It\^o formula from~\cite{GKS-13}.\footnote{ which, however, deals with a larger class of $V\rq{}s$ in this situation}

We have also included important spectral theoretic applications of Theorem \ref{FK}, such as Kato-type inequalities for the ground state energy (Theorem \ref{kato}), $\mathsf{L}^p$-smoothing properties of $\mathrm{e}^{-tH_{\Phi,V}}$ (Theorem \ref{lp}), and integral kernel estimates (Proposition \ref{kern}). We would particularly like to mention Theorem \ref{godd}, where we derive a Golden-Thompson type inequality for $\mathrm{e}^{-tH_{\Phi,V}}$, which seems to be new in its stated form in a both \lq\lq{}noncompact\rq\rq{} and vector-valued setting.

We close this section with an outline of the organization of our paper. In sections~\ref{SS:setting} and~\ref{SS:v-bundle} we describe the setting of discrete weighted graphs and the notion of (Hermitian) vector bundle and (unitary) connection. The necessary operator theoretic preliminaries are explained in section~\ref{SS:operator}, where a precise quadratic form definition of $H_{\Phi,V}$ is given (the reader may find various essential self-adjointness results related with operators of the form $H_{\Phi,V}$ in~\cite{MilTu}; see also ~\cite{Guneysu-Post-12,Grummt-Kolb} for an interplay of probability theory and essential self-adjointness in the manifold situation). The underlying probabilistic concepts are presented in section~\ref{SS:prob-prelim}, together with  the main result of the paper, the Feynman-Kac formula (Theorem~\ref{FK}). The above mentioned applications of the latter are contained in section~\ref{SS:applications}. The last part of the paper (section~\ref{beweis}) is devoted to the proof of Theorem~\ref{FK}, and finally, we have included an appendix, where some useful identities and inequalities on path ordered exponentials have been collected.

\section{Main results}\label{S:main}

\subsection{Setting}\label{SS:setting}

We fix a \emph{graph} $(X,b)$, that is, $X$ is a countable set and $b\colon X\times X\to [0,\infty)$ is a symmetric function such that for all $x\in X$ one has $b(x,x)=0$ and
$$
\sum_{y\in X} b(x,y)<\infty.
$$
 We recall that in this context, the elements of $X$ are called the \emph{vertices} of $(X,b)$, the elements $x,\, y\in X$ with $b(x, y) > 0$ are called \emph{neighbors} of $(X,b)$, and this relationship is denoted by $x\sm y$. We assume that $(X,b)$ is \emph{connected}, that is, for any two vertices $x,\,y$ there exists a chain of vertices $x_1,\,x_2,\,\dots,x_n$ such that $x=x_1$, $y=x_n$, and $x_{j}\sm x_{j+1}$ for all $1\leq j\leq n-1$. The complex algebras of complex-valued and complex-valued finitely supported functions on $X$ will be denoted by $\mathsf{C}(X)$ and $\mathsf{C}_c(X)$, respectively. \vspace{1mm}

We furthermore fix an arbitrary function $m:X\to (0,\infty)$ and recall that the triple $(X,b,m)$ is referred to as a (connected) \emph{weighted graph}. Any such $m$ determines a measure $\mu_m$ on the discrete space $X$ through
$$
\mu_m(A):=\sum_{x\in A} m(x)\>\text{ for any $A\subset X$.}
$$
In particular, we get the corresponding complex-valued $L^p$ spaces which will be denoted
by $\ell^p_m (X)$ for any $p\in[1,\infty]$, where obviously $\ell^{\infty}(X):=\ell^{\infty}_m(X)$ does not depend on $m$.

We will also use the notations
\[
\textrm{deg}_{1,b}(x):=\sum_{y\in X}b(x,y),\qquad  \textrm{deg}_{m,b}(x):=\frac{1}{m(x)}\sum_{y\in X}b(x,y),\qquad x\in X.
\]

\subsection{Vector bundles and connections on discrete structures}\label{SS:v-bundle}

The following definitions are central to this paper. We recall the definition of vector bundles and connections over discrete structures:

\begin{Definition}\label{vec} (i) A family of complex linear spaces $F=\bigsqcup_{x\in X}F_x$ is called a (finite-dimensional) \emph{complex vector bundle over $X$} and written $F\to X$, if there is a $\nu\in\IN$ such that for any $x\in X$ one has $F_x\cong \IC^{\nu}$ as complex vector spaces. Then $\mathrm{rank}(F):=\nu$ is called the \emph{rank} of $F\to X$, the $F_x$\rq{}s are called the \emph{fibers} of $F\to X$, and the complex linear space
$$
\Gamma(X,F):=\prod_{x\in X} F_x=\left.\Big\{f\right| f:X\to F, f(x)\in F_x\Big\}
$$
is called the space of \emph{sections} in $F\to X$. We also define the corresponding linear space of \emph{finitely supported sections} $\Gamma_c(X,F)$ in $F\to X$ to be the set of $f\in \Gamma(X,F)$ such that $f(x)=0$ for all but finitely many $x\in X$.\vspace{1mm}

(ii) Let $F\to X$ be a complex vector bundle. An assignment $\Phi$ which assigns to any $x\sm y$ an isomorphism of complex vector spaces $\Phi_{x,y}: F_x\to F_y$ is called a \emph{$b$-connection} on $F\to X$, if one has $\Phi_{y,x}=\Phi^{-1}_{x,y}$ for all $x\sm y$.
\end{Definition}

\begin{Remark} 1. If $F\to X$ is a complex vector bundle, then so are
\begin{align}
& F^*:=\bigsqcup_{x\in X}F_x^*\longrightarrow X,\nn\\
&\mathrm{End}(F):=\bigsqcup_{x\in X}\mathrm{End}(F_x)\longrightarrow X, \nn\\
&F\boxtimes F^*:=\bigsqcup_{(x,y)\in X\times X}F_x\otimes F_y^*=\bigsqcup_{(x,y)\in X\times X}\mathrm{Hom}(F_y,F_x)\longrightarrow X\times X,\nn
\end{align}
where $F_x^*$ stands for the dual space. More generally, all functorial constructions of new complex vector spaces induce new constructions of complex vector bundles over $X$ in the obvious way.

2. The spaces $\Gamma(X,F)$ and $\Gamma_c(X,F)$ become $\mathsf{C}(X)$-modules by pointwise multiplication.
\end{Remark}

We continue with the notion of Hermitian vector bundles and their unitary connections on discrete sets:

\begin{Definition}\label{herm} (i) A family of complex scalar products
$$
(\bullet,\bullet)_x:F_x\times F_x\longrightarrow \IC, \>x\in X,
$$
is called a \emph{Hermitian structure} on the complex vector bundle $F\to X$, and the pair given by $F\to X$ and $(\bullet,\bullet)_{\bullet}$ is called a \emph{Hermitian vector bundle} over $X$. \vspace{1mm}

(ii) Let $F\to X$ be a complex vector bundle with a $b$-connection $\Phi$ defined on it. Then $\Phi$ is called \emph{unitary} with respect to a Hermitian structure $(\bullet,\bullet)_{\bullet}$ on $F\to X$, if for all $x\sm y$ one has $\Phi_{x,y}^*=\Phi^{-1}_{x,y}$ with respect to $(\bullet,\bullet)_{\bullet}$.
\end{Definition}

\begin{Remark} 1. In the situation of Definition \ref{herm}.(i), the norm and the operator norm on $F_x$ corresponding to $(\bullet,\bullet)_x$ will be denoted by $|\bullet|_x$, and $|\bullet|_{x,y}$ will stand for the operator norm on $\mathrm{Hom}(F_x,F_y)$, so that $|\bullet|_{x,x}=|\bullet|_{x}$.

2. We will stick to the usual abuse of notation and 
suppress the underlying Hermitian structure of a given Hermitian vector bundle over $X$ from now on. This will cause no danger of confusion.

\end{Remark}

\subsection{Operators and semigroups under consideration}\label{SS:operator}

Throughout this section, we fix a Hermitian vector bundle $F\to X$. \vspace{1mm}

Generalizing the scalar $\ell^p_m(X)$ spaces, we get for any $q\in [1,\infty)$ the corresponding complex Banach space $\Gamma_{\ell^q_m}(X,F)$ of $\ell^q_m$-sections in $F\to X$ as follows:

A section $f\in \Gamma(X,F)$ is in $\Gamma_{\ell^q_m}(X,F)$, if and only if
$$
\left\|f\right\|_{m;q}:=\left(\sum_{x\in X} |f(x)|^{q}_xm(x)\right)^{\f{1}{q}}<\infty.
$$
In particular, $\Gamma_{\ell^2_m}(X,F)$ becomes a complex Hilbert space, whose scalar product will be denoted by $\left\langle\bullet,\bullet\right\rangle_{m}$, and we use the convention
$\left\|\bullet\right\|_{m}:=\left\|\bullet\right\|_{m;2}$. We also have the complex Banach space $\Gamma_{\ell^{\infty}}(X,F)$ of bounded sections in $F\to X$, with norm
$$
\> \left\|f\right\|_{\infty}:=\sup_{x\in X} |f(x)|_x.
$$
It is clear that $\Gamma_{c}(X,F)$ is dense in $\Gamma_{\ell^q_m}(X,F)$ for all $q\in [1,\infty)$. The corresponding operator norms will also be denoted by   $\left\|\bullet\right\|_{m;q}$ and $\left\|\bullet\right\|_{\infty}$, and the corresponding operator norm on the space of operators from $\ell^q_m$-sections to $\ell^{q\rq{}}_m$-sections will be denoted by $\left\|\bullet\right\|_{m;q,q\rq{}}$.\vspace{1mm}

Next, we fix a unitary $b$-connection $\Phi$ on $F\to X$. \vspace{1mm}

Then we can define a sesqui-linear form $Q_{\Phi,0}\rq{}$ in $\Gamma_{\ell^2_m}(X,E)$ with domain of definition $\Gamma_{c}(X,F)$ by setting
\begin{align}
Q_{\Phi,0}\rq{}(f_1,f_2)=&\>\frac{1}{2}\sum_{x\sm y}b(x,y)\Big( f_1(x)-\Phi_{y,x} f_1(y),f_2(x)-\Phi_{y,x} f_2(y)\Big)_{x}.\nn
\end{align}
It follows that $Q_{\Phi,0}\rq{}$ is densely defined, symmetric, closable, and nonnegative, and its closure will be denoted by $Q_{\Phi,0}$. We continue with the following:

\begin{Definition}\label{pot} A section $V\in \Gamma(X,\mathrm{End}(F))$ is called a \emph{potential} on $F\to X$, if the operator $V(x):F_x\to F_x$ is self-adjoint for all $x\in X$.
\end{Definition}

In the above situation, we will write $V\geq w$ for a function $w:X\to\IR$, if one has
$$
(V(x)f,f)_x\geq w(x)|f|^2_x \>\>\text{ for all $x\in X$, $f\in F_x$.}
$$

We also fix an arbitrary potential $V$ on $F\to X$ in the sequel. \vspace{1mm}

It 
is easily seen that any such $V$ defines a closed symmetric sesqui-linear form $Q_V$ in $\Gamma_{\ell^2_m}(X,F)$ with
$$
\mathsf{D}(Q_V)=\left.\Big\{f\right|f\in \Gamma_{\ell^2_{m}}(X,F),\,(V(\bullet)f(\bullet),f(\bullet))_{{\bullet}}\in \ell^1_{m}(X)\Big\}
$$
through
$$
Q_V(f_1,f_2)=\sum_{x\in X} ( V(x)f_1(x),f_2(x))_{x}m(x).
$$

Whenever $V$ admits a decomposition $V=V^{+}-V^{-}$ into potentials $V^{\pm}\geq 0$ such that $Q_{V^{-}}$ is $Q_{\Phi,0}$-bounded with bound $<1$, then the KLMN-theorem
implies that the 
sesqui-linear form $Q_{\Phi,V}:=Q_{\Phi,0}+Q_{V}$ with $\mathsf{D}(Q_{\Phi,V})=\mathsf{D}(Q_{\Phi,0})\cap \mathsf{D}(Q_{V})$ is densely defined, symmetric, and semi-bounded\footnote{from now on, \lq\lq{}semi-bounded\rq\rq{} will aways refer to \lq\lq{}semibounded from below\rq\rq{}}. Note that in the above situation we have $\mathsf{D}(Q_{\Phi,V})=\mathsf{D}(Q_{\Phi,0})\cap \mathsf{D}(Q_{V^{+}})$. The self-adjoint semi-bounded operator corresponding to $Q_{\Phi,V}$ will be denoted with $H_{\Phi,V}$, and motivated by \cite{Si0,G2}, we will call
$$
(\mathrm{e}^{-t H_{\Phi,V}})_{t\geq 0}\subset \ILL(\Gamma_{\ell^2_m}(X,F))
$$
the \emph{generalized Schr\"odinger semigroup} corresponding to $\Phi$ and $V$, where the notion \lq\lq{}generalized\rq\rq{} refers to the possibly vector-valued character of our setting.

\begin{Remark}\label{sc} 1. As a particular case of the above construction, we can take $F_x=\{x\}\times\IC$ with its canonic Hermitian structure. Then the sections
 in $F\to X$ can be canonically identified with complex-valued functions on $X$, and under this identification, any $\Phi$ can uniquely be written as $\Phi(x,y)=\mathrm{e}^{\mathrm{i}\theta(x,y)}$, where $\theta$ is a magnetic potential on $(X,b)$, that is, an antisymmetric function on the edges of $(X,b)$ which takes values in $[-\pi,\pi]$. Furthermore, any potential can be uniquely identified with a function $v:X\to \IR$. Thus, as a particular case of the above construction, we get the corresponding quadratic forms $Q_{\theta,v}$ (for appropriate $v$\rq{}s) in $\ell^2_m(X)$. The semibounded self-adjoint operator given by $Q_{\theta,v}$ will be denoted with $H_{\theta,v}$. This situation has been considered in \cite{GKS-13}.

2. The regular Dirichlet form $Q:=Q_{0,0}$ in $\ell^2_m(X)$ and the associated operator $H:=H_{0,0}$ will play a distinguished role in the sequel, for they generate the underlying free scalar evolution (cf.~formula (\ref{prob}) below). We shall denote with
$$
[0,\infty)\times X\times X\longrightarrow (0,\infty),\>(t,x,y)\longmapsto p(t,x,y)
$$
the integral kernel corresponding to $(\mathrm{e}^{-t H})_{t\geq 0}\subset \ILL(\ell^2_m(X))$, that is,
\begin{align}
\mathrm{e}^{-t H}f(x)=\sum_{y\in X} p(t,x,y) f(y) m(y)\>\text{ for all $t\geq 0$, $f\in \ell^2_m(X)$, $x\in X$.} \label{eded}
\end{align}
As $Q$ is a Dirichlet form, $\mathrm{e}^{-t H}$ also operates as
$$
(\mathrm{e}^{-t H})_{t\geq 0}\subset \ILL(\ell^q_m(X)), \>q\in [1,\infty].
$$
In fact, as we have
$$
\sum_{y\in X}p(t,x,y)m(y)\leq 1,
$$
formula (\ref{eded}) holds verbatim for any $f\in \ell^q_m(X)$.
\end{Remark}

\subsection{Probabilistic preliminaries and the Feynman-Kac formula}\label{SS:prob-prelim}

Let us start by recalling a convenient construction of a right Markov process $\IMM$ with values in $X$ (equipped with the discrete topology) which is associated to $Q$. We follow~\cite[Section 3]{GKS-13} very closely here, where the reader may also find some more details on and references about the construction of $\IMM$.

To this end, let $(\Omega,\mathscr{F},\mathbb{P})$ be a fixed probability space, 
  let $\mathbb{N}_{0}$ be the set $\{0,1,2,\dots\}$, and let $(Y_n)_{n\in\mathbb{N}_0}$ be a discrete time Markov chain with values in $X$ such that
\begin{align}
\mathbb{P}(Y_n=x|Y_{n-1}=y)=\frac{b(x,y)}{\textrm{deg}_{1,b}(x)}\qquad \textrm{for all }n\in\IN.\label{trans}
\end{align}
Let $(\xi_n)_{n\in\mathbb{N}_{0}}$ be a sequence of independent exponentially distributed random variables with parameter $1$, which are independent of $(Y_n)_{n\in\mathbb{N}_{0}}$. For $n\in\IN$ we define a sequence of stopping times
\[
S_n:=\frac{1}{\textrm{deg}_{m,b}(Y_{n-1})}\xi_n,\qquad \tau_n:=S_1+S_2+\dots+S_n,
\]
where $\tau_0:=0$. Next, we define the predictable stopping time
$$
\tau:=\sup\{\tau_n|n\in\mathbb{N}_{0} \}  >0,
$$
so that finally, we arrive at the maximally defined, right-continuous process
$$
\mathbb{X}\colon [0,\tau)\times \Omega\longrightarrow  X,\>\>
\mathbb{X}|_{[\tau_n,\tau_{n+1})\times \Omega}:=Y_n,\>\>n\in\mathbb{N}_{0},
$$
which has the $\tau_n$\rq{}s as its jump times and the $S_n$\rq{}s as its holding times. If we define a family of probability measures on $\IFF$ and a filtration of $\IFF$, respectively, by
$$
\mathbb{P}^{x}:=\mathbb{P}(\bullet|\, \mathbb{X}_{0}=x),\>\> \IFF_{t}:=\sigma(\mathbb{X}_{s}|\, s\leq t),
$$
then
\[
\IMM:=(\Omega, \IFF,  (\IFF_{t})_{t\geq 0}, \mathbb{X}, (\mathbb{P}^x)_{x\in X})
\]
is a reversible strong Markov process.

\begin{Remark} 1. $\IMM$ is associated to $Q$ in the sense that
\begin{align}
p(t,x,y)m(y)=\mathbb{P}^x(\mathbb{X}_{t}=y)\>\text{ for all $t\geq 0$, $x,y\in X$.}\label{prob0}
\end{align}
In particular, for any
$q\in [1,\infty]$ and any $f\in\ell^q_{m}(X)$ one has
\begin{align}
\mathrm{e}^{-t H}f(x)=\mathbb{E}^x\left[1_{\{t<\tau\}} f(\mathbb{X}_{t})\right].\label{prob}
\end{align}
2. It follows easily from (\ref{trans}) and the definition of $\mathbb{X}$ that
\begin{align}
\mathbb{P}\Big(b(\mathbb{X}_{\tau_n},\mathbb{X}_{\tau_{n+1}}) >0\text{ for all $n\in\IN_0$}\Big) =1\label{jump}
\end{align}
which precisely means that $\mathbb{X}$ can only jump to neighbors $\mathbb{P}$-a.s.
\end{Remark}

Next, let us introduce the process
$$
N(t):=\sup\{n\in\IN_0 | \tau_n\leq t\}:\Omega\longrightarrow \IN_0\cup \{\infty\},
$$
which at a fixed time $t\geq 0$ counts the jumps of $\mathbb{X}$ until $t$, so that
\begin{align}
\{N(t)<\infty\}=\{t<\tau\}\>\text{ for all $t\geq 0$.}\label{bez}
\end{align}

The following notation will be useful in the sequel: If $\{A_1,\dots,A_n\}$ is a set of linear maps, then, whenever it makes sense, we will write
$$
\prod^{\longleftarrow}_{1\leq j\leq n} A_j:=A_n\cdots A_1,
$$
that is, the $A_j$'s are being ordered decreasingly with respect to their indices.

There is a canonic process which is induced by $\Phi$:

\begin{Definition} The process
\begin{align*}
&\pa^{\Phi}: [0,\tau)\times \Omega\longrightarrow F \boxtimes F^*\\
&\pa^{\Phi}_t:=
\begin{cases}
\mathbf{1}_{\mathbb{X}_0}:=\mathbf{1}_{F_{\mathbb{X}_0}},\> \text{ if $N(t)=0$}\\ \\
\prod^{\longleftarrow}_{1\leq j\leq N(t) } \Phi_{\mathbb{X}_{\tau_{j-1}},\mathbb{X}_{\tau_{j}}}\>\text{ else}
\end{cases}\in \mathrm{Hom}(F_{\mathbb{X}_0},F_{\mathbb{X}_t})
\end{align*}
is called the \emph{$\Phi$-parallel transport} along the paths of $\mathbb{X}$.
\end{Definition}

Note that $\pa^{\Phi}$ is well-defined by (\ref{jump}) and (\ref{bez}) (remember here that $\Phi$ is only defined on neighbors) and that for each $t\geq 0$ and any sample $\omega\in \{t<\tau\}$, the operator
$$
\pa^{\Phi}_t(\omega)\in \mathrm{Hom}(F_{\mathbb{X}_0(\omega)},F_{\mathbb{X}_t(\omega)})\>\text{ is unitary},
$$
a property that will be used implicitly and that will turn out to be central in what follows. We will denote its inverse by $//_{t}^{\Phi,-1} (\omega)$.

Let us now take the potential $V$ into account. We define the self-adjoint process
\begin{align*}
&V^{\Phi}:[0,\tau)\times\Omega\longrightarrow \mathrm{End}(F),\\
&V^{\Phi}_t(\omega):=\left.-//^{\Phi,-1}_tV(\mathbb{X}_t)//^{\Phi}_t\right|_{\omega}\in \mathrm{End}(F_{\mathbb{X}_0(\omega)}).
\end{align*}
Let $\omega\in\Omega$, $0\leq t\leq \tau(\omega)$. As we have
\begin{align}
\int^t_0\left|V^{\Phi}_s(\omega)\right|_{\mathbb{X}_{0}(\omega)}\Id s\leq\sum^{N(t)}_{j=0} \left|V(\mathbb{X}_{\tau_j})\right|_{\mathbb{X}_{\tau_j}}(\tau_{j+1}-\tau_j)\mid_{\omega},\nn
\end{align}
formula (\ref{bez}) implies
\begin{align}
V^{\Phi}_{\bullet}(\omega)\in \mathsf{L}^1_{\mathrm{loc}}[0,\tau(\omega))\otimes\mathrm{End}(F_{\mathbb{X}_0(\omega)}).\label{loci}
\end{align}

Let us consider the initial value problem

\begin{align}
 \mathscr{V}^{\Phi}_{t}(\omega) &=\mathbf{1}_{\mathbb{X}_0(\omega)}+ \int^t_0\mathscr{V}^{\Phi}_{s}(\omega)V^{\Phi}_s(\omega)\Id s \>\>\text{ in $\mathrm{End}(F_{\mathbb{X}_0(\omega)})$.}\label{int}
\end{align}

Using (\ref{loci}) 
and  Proposition \ref{app1} (i) (see Appendix) we immediately get the following result:

\begin{Proposition}\label{poe} Let $\omega\in \Omega$ and $0\leq t\leq \tau(\omega)$. Then there is a unique solution
$$
\mathscr{V}^{\Phi}_{\bullet}(\omega):[0,\tau(\omega))\longrightarrow \mathrm{End}(F_{\mathbb{X}_0(\omega)})
$$
of (\ref{int}). In fact, it is given by the path ordered exponential
\begin{align}
\mathscr{V}_{t}^{\Phi}(\omega)=\mathbf{1}_{\mathbb{X}_0(\omega)}+\sum_{n=1}^{\infty}\int_{t\sigma_n}V_{s_1}^{\Phi}(\omega)\cdots V_{s_n}^{\Phi}(\omega) \,\Id s_1\,\cdots\,\Id s_n ,\label{poi}
\end{align}
where
\[
t\sigma_n:=\{(s_1,s_2,\dots,s_n)| \>0\leq s_1\leq s_2\leq\cdots\leq s_n\leq t\}
\]
denotes the $t$-scaled standard $n$-simplex in $\IR^n$, and where the series converges absolutely and locally uniformly in $t$. More specifically, one has the following a priori bound\emph{:} For all $n\in \IN$,
\begin{align}
&\int_{t\sigma_n}\left|V_{s_1}^{\Phi}(\omega)\cdots V_{s_n}^{\Phi}(\omega)\right|_{\mathbb{X}_0(\omega)} \,\Id s_1\,\cdots\,\Id s_n \leq \f{\left(\int^t_0 |V^{\Phi}_s(\omega)|_{\mathbb{X}_0(\omega)}\Id s\right)^n }{ n!}.\nn
\end{align}
\end{Proposition}

In view of Proposition \ref{poe}, we have canonically associated the process
$$
\mathscr{V}^{\Phi}:[0,\tau)\times\Omega\longrightarrow \mathrm{End}(F)
$$
with $\Phi$ and $V$, where clearly for any $x\in X$, $t\geq 0$, one has
\begin{align*}
\pa^{\Phi}_t(\omega)\in \mathrm{Hom}(F_{x},F_{\mathbb{X}_t(\omega)}),\> \mathscr{V}^{\Phi}_{t}(\omega)\in\mathrm{End}(F_{x})\>\text{ for $\mathbb{P}^x$-a.e. $\omega\in\{t<\tau\}$}.
\end{align*}

\begin{Remark}\label{sdf} In the scalar situation of Remark \ref{sc}.1, we have (with an obvious notation)
$$
\pa^{\theta}_t=\mathrm{e}^{\mathrm{i}\int^t_0\theta(\Id \mathbb{X}_s )}, \>\text{with }\>\int^t_0\theta(\Id \mathbb{X}_s ):=\sum^{N(t)}_{j=1}\theta(\mathbb{X}_{\tau_{j-1}},\mathbb{X}_{\tau_{j}}): [0,\tau)\times\Omega\to \IR
$$
the stochastic line integral of $\theta$ along $\mathbb{X}$ (cf. \cite{GKS-13}). Thus, as everything commutes, for any potential $v:X\to\IR$ one has that $\mathscr{V}^{\theta}_t=\mathrm{e}^{-\int^t_0v(\mathbb{X}_s)\Id s}$ in fact does not depend on $\theta$. This makes some of the calculations in the scalar situation fairly simple. For instance, the proof of claim (i) from section \ref{SS:pf-fki-finite} becomes trivial in the scalar case, whereas the general case is very technical.
\end{Remark}

With these preparations, the following Feynman-Kac type path integral formula for the generalized Schr\"odinger semigroups under consideration is our main result:

\begin{Theorem}\label{FK} Assume that $V$ admits a decomposition $V=V^{+}-V^{-}$ into potentials $V^{\pm}\geq 0$ such that $Q_{|V^-|}$ is $Q$-bounded with bound $<1$. Then $Q_{V^{-}}$ is $Q_{\Phi,0}$-bounded with bound $<1$, and for any $f\in\Gamma_{\ell^2_m}(X,F)$, $t\geq 0$, $x\in X$ one has
\begin{align}
\mathrm{e}^{-tH_{\Phi,V}}f(x)=\mathbb{E}^x\left[1_{\{t<\tau\}}\mathscr{V}_{t}^{\Phi}//_{t}^{\Phi,-1}f(\mathbb{X}_{t})\right].\label{fds}
\end{align}

\end{Theorem}

From taking Laplace transforms, we directly get:

\begin{Corollary}\label{cor1} In the situation of Theorem \ref{FK}, let $\lambda\in\mathbb{C}$ with $\mathrm{Re}(\lambda)> \min\mathrm{spec}(H_{\Phi,V})$, and let $k\in\IN$. Then one has
\begin{align*}
(H_{\Phi,V}+\lambda)^{-k} f(x)=\f{1}{(k-1)!}\int^{\infty}_0 t^{k-1}\mathrm{e}^{-t\lambda}\mathbb{E}^x\left[1_{\{t<\tau\}}\mathscr{V}_{t}^{\Phi}//_{t}^{\Phi,-1}f(\mathbb{X}_{t})\right]\Id t.
\end{align*}
\end{Corollary}

The proof of Theorem \ref{FK} will be given in Section \ref{beweis}. Here, we just add:

\begin{Remark} 1. The fact that for $V$ as in Theorem \ref{FK}, the form $Q_{V^{-}}$ is $Q_{\Phi,0}$-bounded with bound $<1$ will itself be derived
 from a path integral argument, and thus also fits philosophically into the statement (see also \cite{G1}).

2. Note that by expanding $\mathscr{V}_{t}^{\Phi}$ into a path ordered exponential according to Proposition \ref{poe}, formula (\ref{fds}) induces an intuitive perturbation formula in the obvious way which precisely explains how the interaction of the \lq\lq{}free system\rq\rq{} $H_{\Phi,0}$ with the perturbation $V$ takes place at a fixed time.

3. In the scalar situation (that is, if $F_x=\{x\}\times \IC$), Remark \ref{sdf} implies that the above Feynman-Kac formula reduces to the Feynman-Kac-It\^o formula from Theorem 4.1 in \cite{GKS-13}: If $\theta$ is a magnetic potential and if $v:X\to\IR$ is a suitable electric potential, then for all $f\in\ell^2_m(X)$ one has
$$
\mathrm{e}^{-tH_{\theta,v}}f(x)=\mathbb{E}^x\left[1_{\{t<\tau\}}\mathrm{e}^{-\int^t_0v(\mathbb{X}_{s})\Id s}\mathrm{e}^{\mathrm{i}\int^t_0\theta(\Id \mathbb{X}_s )}f(\mathbb{X}_{t})\right],
$$
so that Theorem \ref{FK} can be considered as a natural generalization of Theorem 4.1 from \cite{GKS-13}. In fact, we shall even use the Feynman-Kac-It\^o formula with $\theta=0$ in the proof of Theorem \ref{FK} through a dominated convergence argument.

4. Comparable vector-bundle-valued path integral formulae for operators acting on Riemannian manifolds can be found in \cite{G-10, G2}.
\end{Remark}

\subsection{Applications of the Feynman-Kac formula}\label{SS:applications}

\subsubsection{Semigroup domination and $\mathsf{L}^q$-estimates.} We start with the following observation on semigroup domination, which is the generalization of Theorem 5.2 in \cite{GKS-13} to our vector bundle setting, and is particularly important, for it allows us to estimate vector-data by certain scalar data:

\begin{Theorem}\label{kato} Under the assumptions of Theorem \ref{FK}, let $w:X\to \IR$ satisfy the following conditions\emph{:}
\begin{enumerate}
\item[$\bullet$] There exist $w^{\pm}:X\to [0,\infty)$ with $w=w^+-w^-$ such that $Q_{w^-}$ is $Q$-bounded with bound $<1$.
\item[$\bullet$] One has $V\geq w$.
\end{enumerate}
Then the following assertions hold\emph{:}
\begin{enumerate}
\item[\emph{(i)}] For all $t\geq 0$, $f\in\Gamma_{\ell^2_m}(X,F)$, $x\in X$, we have
\[
|\mathrm{e}^{-tH_{\Phi,V}}f(x)|_{x}\leq \mathrm{e}^{-tH_{0,w}}|f|(x).
\]

\item[\emph{(ii)}] For all $f\in\mathsf{D}(Q_{\Phi, V})$, we have $|f|\in\mathsf{D}(Q_{0,w})$ and
\[
Q_{\Phi, V}(f,f)\geq  Q_{0,w}(|f|,|f|).
\]

\item[\emph{(iii)}] One has $\min\mathrm{spec}(H_{\Phi,V})\geq \min\mathrm{spec}(H_{0,w})$.

\item [\emph{(iv)}] For all $f\in\Gamma_{\ell^2_m}(X,F)$, $\lambda\in\mathbb{C}$ with $\mathrm{Re}(\lambda)> \min\mathrm{spec}(H_{\Phi,V})$, $k\in\IN$, and all $x\in X$ we have
\[
|(H_{\Phi,V}+\lambda)^{-k}f(x)|_{x}\leq (H_{0,w}+\lambda)^{-k}|f|(x).
\]
\end{enumerate}
\end{Theorem}

\begin{proof} Noting that $\mathrm{e}^{-tH_{0,w}}$ is given by
\begin{align}
\mathrm{e}^{-tH_{0,w}}u(x)=\mathbb{E}^x\left[1_{\{t<\tau\}}\mathrm{e}^{-\int^t_0 w(\mathbb{X}_s)\Id s}u(\mathbb{X}_t)\right],\>u \in \ell^2_m(X),\label{scall}
\end{align}
and being equipped with Proposition \ref{app2} (i), the assertion (i) follows directly from Theorem \ref{FK}. In view of Corollary \ref{cor1}, (i) directly implies (iv).
The other assertions are implied by (i) through purely functional analytic arguments (see for example the proof of Theorem 2.13 in \cite{G2} for these arguments).
\end{proof}

Let us explain how the latter result can be used to derive $\mathsf{L}^p$-estimates for the underlying Schr\"odinger semigroups, with some additional control on $V^-$ and $p(t,x,y)$. To this end, we recall:

\begin{Definition} A function $w:X\to \IC$ is said to be in the \emph{Kato class} $\mathcal{K}(Q)$ corresponding to $Q$, if
$$
\lim_{t\to 0+}\sup_{x\in X}\int^t_0\sum_{y\in X} p(s,x,y) |w(y)|m(y)\Id s=0.
$$
\end{Definition}

Note that for any $w:X\to \IC$, $s\geq 0$ one has
$$
\sum_{y\in X} p(s,x,y) |w(y)| m(y)= \mathbb{E}^x\left[  1_{\{s<\tau\}}\left|w(\mathbb{X}_s)\right|
  \right],
$$
which follows from (\ref{prob}) and which produces an equivalent probabilistic characterization of the Kato class. Furthermore, one obviously has $\ell^{\infty}(X)\subset \mathcal{K}(Q)$, and in practice (where one usually has some bounds of the form $p(t,x,y)\leq f(t)$ for small $t$) one can produce large $\ell^{p}_m(X)$-type subspaces of the Kato class. We refer the reader to \cite{G1,kt} for arguments of this type.

The importance of $\mathcal{K}(Q)$ lies in the following observation:

\begin{Proposition}\label{kat} The following statements hold for any $w\in \mathcal{K}(Q)$\emph{:}

\emph{(i)} The form $Q_{w}$ is $Q$-bounded with bound $<1$.

\emph{(ii)} For any $\delta>1$ there is a $C(w,\delta)>0$ such that for all $t\geq 0$,
\begin{align}
\sup_{x\in X} \mathbb{E}^x\left[1_{\{t<\tau\}}\mathrm{e}^{\int^t_0 \left|w(\mathbb{X}_s)\right|\Id
s}\right] \leq \delta \mathrm{e}^{tC(w,\delta)}. \nn
\end{align}
\end{Proposition}

\begin{proof} (i) This follows from an abstract result on regular Dirichlet forms, \cite[Theorem~3.1]{peter}.

(ii) This assertion follows from the proof of Lemma 3.8 from \cite{pall}, which only uses the Markov property of Brownian motion and thus applies directly in our situation.
\end{proof}

Proposition \ref{kat} can be combined with Theorem \ref{kato} (i) to give the following result:

\begin{Theorem}\label{lp} Assume that $V$ admits a decomposition $V=V^{+}-V^{-}$ into potentials $V^{\pm}\geq 0$ such that $|V^-|\in \mathcal{K}(Q)$. Then for all $\delta>1$ there is a $C(V^-,\delta)>0$ such that for all $q\in [2,\infty]$ and all $t\geq 0$ with
\begin{align}
C(t):=\sup_{x,y\in X}p(t,x,y)<\infty\label{boun}
\end{align}
one has
$$
\left\|\mathrm{e}^{-tH_{\Phi,V}}\right\|_{m;2,q}\leq \delta \mathrm{e}^{tC(|V^-|,\delta)}C(t)^{\f{1}{2}-\f{1}{q}}.
$$
\end{Theorem}

\begin{Remark} Note that in the situation of Theorem \ref{lp}, Proposition \ref{kat} implies that the assumption on $V$ from Theorem \ref{FK} is satisfied, in particular, $H_{\Phi,V}$ is well-defined.
\end{Remark}

\begin{proof}[Proof of Theorem \ref{lp}] Setting
\begin{align}
&w^{+}:=\min\mathrm{spec}(V^{+}(\bullet)):X\to [0,\infty),\nn\\
& w^{-}:=\max\mathrm{spec}(V^{-}(\bullet)):X\to [0,\infty),\nn\\
&w:=w^+-w^-:X\to\IR, \label{sxs}
\end{align}
and using Theorem \ref{kato} (i), we find that it is sufficient to prove
$$
\left\|\mathrm{e}^{-tH_{0,w}}\right\|_{m;2,q}\leq \delta C(t)^{\f{1}{2}-\f{1}{q}}\mathrm{e}^{tC(w^-,\delta)},
$$
a bound on a scalar operator. To this end, note first that $C(t)<\infty$ implies (cf. Proposition 2.18 in \cite{G2} and the remark thereafter)
\begin{align}
\left\|\mathrm{e}^{-tH}\right\|_{m;q\rq{},q\rq{}\rq{}}\leq C(t)^{\f{1}{q\rq{}}-\f{1}{q\rq{}\rq{}}}\text{ for all $q\rq{},q\rq{}\rq{}\in [1,\infty]$ with $q\rq{}\leq q\rq{}\rq{}$},\label{boun2}
\end{align}
and that $\mathrm{e}^{-tH_{0,w}}$ is given by (\ref{scall}).

Case $q=\infty$: For any $u\in \ell^2_m(X)$ one has
\begin{align}
\left\|\mathrm{e}^{-t H_{0,w}}u\right\|_{\infty}&\leq \sup_{x\in X}\mathbb{E}^x \left[  1_{\{t<\tau\}}\mathrm{e}^{-\int^t_0 w(\mathbb{X}_s) \Id s}  | u(\mathbb{X}_t)| \right]\nn\\
& \leq \sup_{x\in X}\mathbb{E}^x \left[ 1_{\{t<\tau\}} \mathrm{e}^{-2\int^t_0 w(\mathbb{X}_s) \Id s}  \right]^{\f{1}{2}}\sup_{x\in X}\mathbb{E}^x \left[ 1_{\{t<\tau\}} | u(\mathbb{X}_t)|^{2} \right]^{\f{1}{2}}\nn\\
& \leq   (\delta \mathrm{e}^{t C(w^-,\delta)})^{\f{1}{2}}\left\|\mathrm{e}^{-tH}|u|^2\right\|^{\f{1}{2}}_{\infty}\nn\\
&\leq (\delta C(t)  \mathrm{e}^{t C(w^-,\delta)} )^{\f{1}{2}}         \left\|u\right\|_{m;2},\nn
\end{align}
where we have used Cauchy-Schwarz inequality for the second inequality, $-w\leq w^-$ in combination with Proposition \ref{kat} (ii) and (\ref{prob}) for the third inequality, and (\ref{boun2}) for the last inequality. \vspace{1mm}

Case $q<\infty$: We set $l:=q/2$. Then for any $u\in \ell^2_m(X)$ the inequalities
\begin{align}
\left\|\mathrm{e}^{-t H_{0,w}}u\right\|^q_{m;q}&\leq \sum_{x\in X} \mathbb{E}^x \left[1_{\{t<\tau\}} \mathrm{e}^{-\int^t_0 w(\mathbb{X}_s) \Id s} |u(\mathbb{X}_t)| \right]^q m(x)\nn\\
& \leq \sum_{x\in X}  \mathbb{E}^x \left[1_{\{t<\tau\}}   \mathrm{e}^{-2\int^t_0 w(\mathbb{X}_s) \Id s}\right]^{l}\mathbb{E}^x \left[ 1_{\{t<\tau\}} | u(\mathbb{X}_t)|^{2}\right]^{l} m(x)\nn\\
& \leq\delta^l\mathrm{e}^{lt C(w^-,\delta)} \sum_{x\in X}  \mathbb{E}^x \left[ 1_{\{t<\tau\}} | u(\mathbb{X}_t)|^{2} \right]^{l} m(x)\nn
\end{align}
follow again from Cauchy-Schwarz inequality, $-w\leq w^-$ and Proposition \ref{kat} (ii). Finally, using (\ref{prob}) and (\ref{boun2}) we arrive at
\begin{align}
\left\|\mathrm{e}^{-t H_{0,w}}u\right\|_{m;q}&\leq (\delta \mathrm{e}^{t C(w^-,\delta)})^{\f{1}{2}} \left\|\mathrm{e}^{-tH}\right\|^{\f{l}{q}}_{m;1,l}\left\||u|^2\right\|^{\f{l}{q}}_{m;1} \nn\\
&= (\delta\mathrm{e}^{t C(w^-,\delta)})^{\f{1}{2}} \left\|\mathrm{e}^{-tH}\right\|^{\f{l}{q}}_{m;1,l}\left\|u\right\|_{m;2}\nn\\
&\leq (\delta\mathrm{e}^{t C(w^-,\delta)})^{\f{1}{2}} C(t)^{\f{1}{2}-\f{1}{q}} \left\|u\right\|_{m;2}.\nn
\end{align}
This completes the proof.
\end{proof}

\subsubsection{Integral kernels and trace estimates.} Let $\IJJ^p(\IHH_1,\IHH_2)$ denote the $p$-th Schatten class of operators between Hilbert spaces $\IHH_1$ and $\IHH_2$, with the usual convention $\IJJ^p(\IHH_1):=\IJJ^p(\IHH_1,\IHH_1)$. In particular, $p=1$ is the trace class, $p=2$ the Hilbert-Schmidt class, and $p=\infty$ the compact class.

Let $\mathrm{tr}_x(\bullet)$ denote the operator trace on $\mathrm{End}(F_x)$. We recall the following simple result:

\begin{Propandef}\label{hilb} The following statements hold for any $T\in \ILL(\Gamma_{\ell^2_m}(X,F))$\emph{:}

\emph{(i)} There is a unique section
$$
T(\bullet,\bullet)\in\Gamma(X\times X, F\boxtimes F^*),
$$
the \emph{integral kernel} of $T$, such that
$$
Tf(x)=\sum_{y\in X}T(x,y) f(y) m(y)\>\text{ for all $f\in\Gamma_{\ell^2_m}(X,F)$, $x\in X$.}
$$
\emph{(ii)} The operator $T$ is self-adjoint if and only if $T(x,y)\in \mathrm{Hom}(F_y,F_x)$ is self-adjoint for any $(x,y)\in X\times X$.

\emph{(iii)} The Hilbert-Schmidt norm of $T$ is given by
$$
\left(\sum_{x\in X}\sum_{y\in X} \mathrm{tr}_y\Big(T(x,y)^*T(x,y)\Big) m(x)m(y)\right)^{\f{1}{2}}\in [0,\infty].
$$
In particular, one has $T\in\IJJ^2\left(\Gamma_{\ell^2_m}(X,F)\right)$ if and only if the latter number is finite.
\end{Propandef}

\begin{proof} The first part follows easily from the discreteness of $X$, (ii) and (iii) are standard.
\end{proof}

In order to derive an explicit probabilistic formula for the integral kernel of $\mathrm{e}^{-tH_{\Phi,V}}$, we consider as in \cite{GKS-13} the conditional probability measure $\mathbb{P}^{x,y}_t$ on $\{t<\tau\}$ defined by
\[
\mathbb{P}^{x,y}_t:=\mathbb{P}^x(\bullet\left| \mathbb{X}_t=y)\right.\>\>\text{for any $x,y\in X$, $t>0$,}
\]
which is well-defined as $m>0$ and as we have $p(t,\bullet,\bullet)>0$ by the connectedness of $(X,b)$. Note that (\ref{prob}) implies
\begin{align}
\mathbb{P}^{x}(A)&=\sum_{y\in X}\mathbb{P}^{x,y}_t(A)\mathbb{P}^x(\mathbb{X}_t=y)
=\sum_{y\in X}\mathbb{P}^{x,y}_t(A)p(t,x,y)m(y)\label{condi}
\end{align}
for any event $A\subset \{t<\tau \}$, so that
$$
\mathsf{L}^1\left(\{t<\tau\},\mathbb{P}^{x}\right)\otimes\IHH\subset \mathsf{L}^1\left(\{t<\tau\},\mathbb{P}^{x,y}_t\right)\otimes\IHH
$$
for any  complex finite dimensional Hilbert space $\IHH$. Let us furthermore note that
\begin{align*}
\pa^{\Phi}_t(\omega)\in \mathrm{Hom}(F_{x},F_{y}),\> \mathscr{V}^{\Phi}_{t}(\omega)\in\mathrm{End}(F_{x})\>\text{ for $\mathbb{P}^{x,y}_t$-a.e. $\omega\in\{t<\tau\}$}.
\end{align*}

Now we easily get:

\begin{Proposition}\label{kern} \emph{(i)} In the situation of Theorem \ref{kato}, for any $t>0$ the integral kernel of $\mathrm{e}^{-tH_{\Phi,V}}$ is given for any $(x,y)\in X\times X$ by
\begin{align}
\mathrm{e}^{-tH_{\Phi,V}}(x,y)&=p(t,x,y)\mathbb{E}^{x,y}_t\left[1_{\{t<\tau\}}\mathscr{V}_{t}^{\Phi}//_{t}^{\Phi,-1}\right]\nonumber\\
&=\f{1}{m(y)}\mathbb{E}^{x}\left[1_{\{\mathbb{X}_t=y\}}\mathscr{V}_{t}^{\Phi}//_{t}^{\Phi,-1}\right];\label{kernn}
\end{align}
in particular,  one has
\begin{align}
\left|\mathrm{e}^{-tH_{\Phi,V}}(x,y)\right|_{y,x}\leq \mathrm{e}^{-tH_{0,w}}(x,y).\label{kerna}
\end{align}

\emph{(ii)} If $V$ admits a decomposition $V=V^{+}-V^{-}$ into potentials $V^{\pm}\geq 0$ such that $|V^-|\in \mathcal{K}(Q)$, then for all $\delta>1$ there is a constant $C(V^-,\delta)>0$ such that for any $t>0$, $(x,y)\in X\times X$ one has
$$\left|\mathrm{e}^{-tH_{\Phi,V}}(x,y)\right|_{y,x}\leq \delta \mathrm{e}^{tC(V^-,\delta)} p(t,x,y).
$$
In particular, under (\ref{boun}), $\mathrm{e}^{-tH_{\Phi,V}}(\bullet,\bullet)$ is bounded in operator norm according to
$$
\sup_{x,y\in X}\left|\mathrm{e}^{-tH_{\Phi,V}}(x,y)\right|_{y,x}\leq \delta \mathrm{e}^{tC(V^-,\delta)} C(t).
$$
\end{Proposition}
\begin{proof} (i) The formulae follow directly from  the Feynman-Kac formula combined with (\ref{condi}) and the definition of $\mathbb{P}^{x,y}_t$, and the inequality is then directly implied by Proposition \ref{app2} (i).

(ii) This is an immediate consequence of (i) and Proposition \ref{kat} (ii).

\end{proof}

We close this section with trace estimates.

\begin{Theorem}\label{godd} The following statements hold for any $t>0$ under the assumptions of Theorem \ref{FK}\emph{:}

\emph{(i)} One has

$$
\mathrm{tr}\left( \mathrm{e}^{-t H_{\Phi,V}}\right)=\sum_{x\in X}p(t,x,x)  \mathrm{tr}_{x}\left(\mathbb{E}^{x,x}_t\left[1_{\{t<\tau\}}\mathscr{V}_{t}^{\Phi}//_{t}^{\Phi,-1}\right]\right)m(x)\in [0,\infty].
$$

\emph{(ii)} Assume that there exists a $w:X\to \IR$ which satisfies the following assumptions\emph{:} One has $V\geq w$, there exist $w^{\pm}:X\to [0,\infty)$ with $w=w^+-w^-$ such that $Q_{w^-}$ is $Q$-bounded with bound $<1$, and it holds that
\begin{align}
\sum_{x\in X} p(t,x,x)\mathrm{e}^{-t w(x)}m(x)<\infty.\nn
\end{align}
Then one has $\mathrm{e}^{-t H_{\Phi,V}}\in \IJJ^1(\Gamma_{\ell^2_m}(X,F))$ with
\begin{align}
\mathrm{tr}\left( \mathrm{e}^{-t H_{\Phi,V}}\right)\leq \sum_{x\in X}p(t,x,x) \mathrm{tr}_x\left(\mathrm{e}^{-t V(x)}\right)m(x)<\infty.\label{gold}
\end{align}
\end{Theorem}

\begin{proof} (i) Combining Proposition \ref{hilb} with standard functional analytic arguments one gets
\begin{align}
\mathrm{tr}\left( \mathrm{e}^{-t H_{\Phi,V}}\right)=\sum_{x\in X}\mathrm{tr}_{x}\left(\mathrm{e}^{-t H_{\Phi,V}}(x,x)\right)m(x)\in [0,\infty],\label{kernb}
\end{align}
so that the asserted formula is implied directly by (\ref{kernn}).

(ii) Let $V_n(x):=\max(-n ,V(x))$. Then we have $-n\leq V_n$ for all $n$, and $V_n\searrow V$ as $n\to\infty$ pointwise in the sense of self-adjoint operators. Applying Theorem B.1 in \cite{GKS-13} with $q\rq{}:=Q_{\Phi,0}$, $q\rq{}\rq{}:=Q_{V_n}$ directly gives the first inequality in
\begin{align}
\mathrm{tr}\left( \mathrm{e}^{-t H_{\Phi,V_n}}\right)&\leq \mathrm{tr}\left( \mathrm{e}^{-\f{t}{2} H_{\Phi,0}}\mathrm{e}^{-t V_n}\mathrm{e}^{-\f{t}{2} H_{\Phi,0}}\right)\nn\\
&=\sum_{x\in X} \mathrm{tr}_x\left(\mathrm{e}^{-t H_{\Phi,0}}(x,x)\mathrm{e}^{-tV_n(x)}\right)m(x)\nn\\
&\leq \sum_{x\in X} \left|\mathrm{e}^{-tH_{\Phi,0}}(x,x)\right|_x\mathrm{tr}_x\left(\mathrm{e}^{-tV_n(x)}\right)m(x)\nn\\
&\leq\sum_{x\in X} p(t,x,x)\mathrm{tr}_x\left(\mathrm{e}^{-tV_n(x)}\right)m(x),\label{hilo}
\end{align}
where we have used
$$
\left(\mathrm{e}^{-\f{t}{2} V_n}\mathrm{e}^{-\f{t}{2} H_{\Phi,0}}\right)(x,y)=\mathrm{e}^{-\f{t}{2} V_n(x)}\mathrm{e}^{-\f{t}{2} H_{\Phi,0}}(x,y)
$$
for the equality, the fact that $\mathrm{e}^{-tV_n(x)}$ is nonnegative for the second equality, and (\ref{kerna}) with zero potentials for the last inequality. Now by the monotonicity of the trace and monotone convergence, we get
\begin{align}
\lim_{n\to\infty}\sum_{x\in X} p(t,x,x)\mathrm{tr}_x\left(\mathrm{e}^{-tV_n(x)}\right)m(x)&=\sum_{x\in X} p(t,x,x)\mathrm{tr}_x\left(\mathrm{e}^{-tV(x)}\right)m(x)\nn\\
&\leq \mathrm{rank}(F)\sum_{x\in X}p(t,x,x) \mathrm{e}^{-t w(x)}m(x)\nn\\
&<\infty.\nn
\end{align}
On the other hand, part (i) of this Theorem implies (with an obvious notation)
$$
\mathrm{tr}\left( \mathrm{e}^{-t H_{\Phi,V_n}}\right)=\sum_{x\in X}p(t,x,x) \ \mathrm{tr}_{x}\left(\mathbb{E}^{x,x}_t\left[1_{\{t<\tau\}}\mathscr{V}_{n,t}^{\Phi}//_{t}^{\Phi,-1}\right]\right)m(x).
$$
In view of $V_n\geq  w$, Proposition \ref{app2} gives
\begin{align}
&\left|p(t,x,x)  \mathrm{tr}_{x}\left(\mathbb{E}^{x,x}_t\left[1_{\{t<\tau\}}\mathscr{V}_{n,t}^{\Phi}//_{t}^{\Phi,-1}\right]\right)\right|_x\nn\\
&\leq \mathrm{rank}(F)p(t,x,x)\mathbb{E}^{x,x}_t\left[1_{\{t<\tau\}}\mathrm{e}^{-\int^t_0w(\mathbb{X}_s)\Id s}\right],\nn
\end{align}
but Theorem 5.3 in \cite{GKS-13} implies $\mathrm{e}^{-t H_{0,w}}\in\IJJ^1(\ell^2_m(X))$ with
\begin{align}
\sum_{x\in X}p(t,x,x)\mathbb{E}^{x,x}_t\left[1_{\{t<\tau\}}\mathrm{e}^{-\int^t_0w(\mathbb{X}_s)\Id s}\right]m(x)=\mathrm{tr}\left(\mathrm{e}^{-t H_{0,w}}\right).\label{domm}
\end{align}
Next, using Proposition \ref{app2} (i) and (ii), respectively, we get
$$
\lim_{n\to\infty}1_{\{t<\tau\}}\mathscr{V}_{n,t}^{\Phi}//_{t}^{\Phi,-1}=1_{\{t<\tau\}}\mathscr{V}_{t}^{\Phi}//_{t}^{\Phi,-1},\>\text{ $\mathbb{P}^{x,x}_t$-a.s.}
$$
and
$$
\left|1_{\{t<\tau\}}\mathscr{V}_{n,t}^{\Phi}//_{t}^{\Phi,-1}\right|_x\leq 1_{\{t<\tau\}}\mathrm{e}^{-\int^t_0w(\mathbb{X}_s)\Id s},\>\text{ $\mathbb{P}^{x,x}_t$-a.s.,}
$$
where
$$
\mathbb{E}^{x,x}_t\left[1_{\{t<\tau\}}\mathrm{e}^{-\int^t_0w(\mathbb{X}_s)\Id s}\right]<\infty
$$
follows from the finiteness of (\ref{domm}). The latter considerations and dominated convergence give
$$
\lim_{n\to\infty}\mathbb{E}^{x,x}_t\left[1_{\{t<\tau\}}\mathscr{V}_{n,t}^{\Phi}//_{t}^{\Phi,-1}\right]=\mathbb{E}^{x,x}_t\left[1_{\{t<\tau\}}\mathscr{V}_{t}^{\Phi}//_{t}^{\Phi,-1}\right],
$$
thus $\mathrm{tr}\left( \mathrm{e}^{-t H_{\Phi,V_n}}\right)\to \mathrm{tr}\left( \mathrm{e}^{-t H_{\Phi,V}}\right)$ as $n\to\infty$, and the proof is complete.
\end{proof}


The inequality (\ref{gold}) is of Golden-Thompson type (cf. \cite{Si05} for scalar Schr\"odinger operators in the Euclidean $\IR^l$, \cite{GKS-13} for scalar Schr\"odinger operators on discrete graphs, and \cite{Baer} for covariant Schr�dinger operators on closed manifolds). In order to derive such a bound in our setting, we could have also simply combined Theorem 5.3 in \cite{GKS-13} with (\ref{kerna}), which would have immediately produced the inequality
\begin{align}
\mathrm{tr}\left( \mathrm{e}^{-t H_{\Phi,V}}\right)\leq \mathrm{rank}(F)\sum_{x\in X}p(t,x,x) \mathrm{e}^{-t w(x)}m(x).\label{inee}
\end{align}
However, the latter bound is obviously much weaker than (\ref{gold}), the proof of which required a functional analytic localization argument. Nevertheless, (\ref{inee}) itself is also important, for it clarifies how Golden-Thompson type inequalities on vector bundles scale with respect to $\mathrm{rank}(\bullet)$.

\section{Proof of Theorem \ref{FK}}\label{beweis}

\subsection{Proof of Theorem \ref{FK} for finite subgraphs}\label{SS:pf-fki-finite}

Let $U\subset X$ be a finite set. In what follows, let $Q^{(U)}_{\Phi,V}$ be the symmetric sesqui-linear form in the finite dimensional Hilbert space
$$
\Gamma_{\ell^2_m}(U,F)  = \Gamma(U,F)=\Gamma_c(U,F),
$$
given by
\begin{align*}
Q^{(U)}_{\Phi,V}(f_1,f_2)= \>&\frac{1}{2}\sum_{\overset{x,y\in U}{ x\sm y}}b(x,y)\Big( f_1(x)-\Phi_{y,x} f_1(y),f_2(x)-\Phi_{y,x} f_2(y)\Big)_{x}\\
& + \sum_{x \in U} (V(x)f_1(x), f_2(x))_xm(x).
\end{align*}
The self-adjoint operator corresponding to $Q^{(U)}_{\Phi,V}$ will be denoted with $H^{(U)}_{\Phi,V}$. For $f\in\Gamma_{\ell^2_m}(U,F)$ define
\[
T^{(U)}_{t}f(x):=\mathbb{E}^x\left[1_{\{t<\tau_{U}\}}\mathscr{V}^{\Phi}_{t}//_{t}^{\Phi,-1}f(\mathbb{X}_{t})\right],\qquad x\in U,\quad t\geq 0,
\]
where $\tau_{U}<\tau$ denotes the first exit time of $\mathbb{X}$ from $U$. We are going to prove the following two assertions in this subsection:\vspace{1.2mm}

Claim (i): \emph{$(T^{(U)}_{t})_{t\geq 0}$ is a strongly continuous self-adjoint semigroup of (bounded) operators on $\Gamma_{\ell^2_m}(U,F)$.} \vspace{1.2mm}

Claim (ii): \emph{For all $f\in\Gamma_{\ell^2_m}(U,F)$, $x\in U$, we have}
$$
\lim_{t\to 0+}\frac{T^{(U)}_{t}f(x)-f(x)}{t}=-H^{(U)}_{\Phi,V}f(x).
$$

Then, as $U$ is a finite set, the formula from claim (ii) will in fact hold in the sense of $\Gamma_{\ell^2_m}(U,F)$, so that claim (i) will imply
\begin{align}
\mathrm{e}^{-t H^{(U)}_{\Phi,V}}f(x)=\mathbb{E}^x\left[1_{\{t<\tau_{U}\}}\mathscr{V}_{t}^{\Phi}//_{t}^{\Phi,-1}f(\mathbb{X}_{t})\right],\label{local}
\end{align}
for all $t\geq 0$, $f\in\Gamma_{\ell^2_m}(U,F)$, $x\in U$, which is the Feynman-Kac formula in this situation.\vspace{1.2mm}

\emph{Proof of claim} (i): In view of the finiteness of $U$, it is a trivial fact that $T^{(U)}_{t}\in \ILL(\Gamma_{\ell^2_m}(U,F))$. The proofs of the semigroup
property and the self-adjointness
property, respectively, become rather technical when compared with the scalar situation, as the involved processes do not commute anymore in general. To this end,
 let $r,s\geq 0$, $f,g\in\Gamma_{\ell^2_m}(U,F)$, $x\in U$.

In order to prove the semigroup property, we define the parallel transport
\begin{align*}
&\pa^{\Phi,r}: [0,\tau)\times \Omega\longrightarrow F \boxtimes F^*\\
&\pa^{\Phi,r}_{t}:=
\begin{cases}
&\mathbf{1}_{\mathbb{X}_r}, \>\text{ if $N(r+t)-N(r)=0$}\\\\
&\prod^{\longleftarrow}_{1\leq j\leq N(r+t)-N(r)}\Phi_{\mathbb{X}_{\tau_{j+N(r)-1}},\mathbb{X}_{\tau_{j+N(r)}}},\>\text{ else}
\end{cases}\in \mathrm{Hom}(F_{\mathbb{X}_r},F_{\mathbb{X}_{r+t}})
\end{align*}
along the $r$-shifted paths of $\mathbb{X}$ and furthermore the corresponding path ordered exponential
\begin{align*}
\mathscr{V}^{\Phi,r}:[0,\tau)\times\Omega\longrightarrow \mathrm{End}(F),\>\mathscr{V}^{\Phi,r}_t\in \mathrm{End}(F_{\mathbb{X}_r}),
\end{align*}
which is given by the unique pathwise solution of the initial value problem
\begin{align}
 \mathscr{V}^{\Phi,r}_{t} &=\mathbf{1}_{\mathbb{X}_r}- \int^t_0\mathscr{V}^{\Phi,r}_{\sigma}\pa^{\Phi,r,-1}_\sigma V(\mathbb{X}_{r+\sigma})\pa^{\Phi,r}_\sigma\Id
\sigma \>\>\text{ in $\mathrm{End}(F_{\mathbb{X}_r})$.}\nn
\end{align}
Using  $\pa^{\Phi}_{r+s}=\pa^{\Phi,r}_s\pa^{\Phi}_r$, one verifies that the two processes
$$
\mathscr{V}^{\Phi}_{r+\bullet}, \>\>\>\mathscr{V}^{\Phi}_{r}\pa^{\Phi,-1}_{r}\mathscr{V}^{\Phi,r}_{\bullet}\pa^{\Phi}_{r}
$$
both solve the same initial value problem, so that by a uniqueness argument we get the formula
$$
\mathscr{V}^{\Phi}_{r+s}\pa^{\Phi,-1}_{r+s}=\mathscr{V}^{\Phi}_{r}\pa^{\Phi,-1}_{r}\mathscr{V}^{\Phi,r}_{s}\pa^{\Phi,r,-1}_{s}.
$$
Now using the Markov property of $\IMM$ we can calculate

\begin{align*}
T^{(U)}_{r+s}f(x)&=\mathbb{E}^x\left[1_{\{r+s<\tau_U\}}\mathscr{V}^{\Phi}_{r}\pa^{\Phi,-1}_{r}\mathscr{V}^{\Phi,r}_{s}\pa^{\Phi,r,-1}_{s}f(\mathbb{X}_{r+s})\right]\\
&=\mathbb{E}^x\left[1_{\{r<\tau_U\}}\mathscr{V}^{\Phi}_{r}\pa^{\Phi,-1}_{r}\mathbb{E}^x\left[\left.1_{\{s<\tau_U\}}\mathscr{V}^{\Phi,r}_{s}\pa^{\Phi,r,-1}_{s}f(\mathbb{X}_{r+s})\right|\mathscr{F}_r\right]\right]\\
&=\mathbb{E}^x\left[1_{\{r<\tau_U\}}\mathscr{V}^{\Phi}_{r}\pa^{\Phi,-1}_{r}\mathbb{E}^{\mathbb{X}_r}\left[1_{\{s<\tau_U\}}\mathscr{V}^{\Phi}_{s}\pa^{\Phi,-1}_{s}f(\mathbb{X}_{s})\right]\right]\\
&=T^{(U)}_{r}T^{(U)}_{s}f(x).
\end{align*}

For the self-adjointness property, we define the parallel transport
\begin{align*}
&\pa^{\Phi,(r)}: [0,r]\times \{r<\tau\}\longrightarrow F \boxtimes F^*,\\
&\pa^{\Phi,(r)}_{t}:=\begin{cases}
&\mathbf{1}_{\mathbb{X}_r}, \>\text{ if $N(r)-N(r-t)=0$}\\\\
&\prod^{\longleftarrow}_{1\leq j\leq N(r)-N(r-t)}\Phi_{\mathbb{X}_{\tau_{N(r)-j+1}},\mathbb{X}_{\tau_{N(r)-j}}},\>\text{ else}
\end{cases}\in \mathrm{Hom}(F_{\mathbb{X}_r},F_{\mathbb{X}_{r-t}})
\end{align*}
along the $r$-reversed paths of $\mathbb{X}$, and the corresponding path ordered exponential
\begin{align*}
\mathscr{V}^{\Phi,(r)}: [0,r]\times \{r<\tau\}\longrightarrow \mathrm{End}(F),\>\mathscr{V}^{\Phi,(r)}_t\in \mathrm{End}(F_{\mathbb{X}_r}),
\end{align*}
which is given as the unique pathwise solution of the initial value problem
\begin{align}
 \mathscr{V}^{\Phi,(r)}_{t} &=\mathbf{1}_{\mathbb{X}_r}- \int^t_0\mathscr{V}^{\Phi,(r)}_{\sigma}\pa^{\Phi,(r),-1}_\sigma V(\mathbb{X}_{r-\sigma})\pa^{\Phi,(r)}_\sigma \Id \sigma \>\>\text{ in $\mathrm{End}(F_{\mathbb{X}_r})$.}
\end{align}
Then we have $\pa^{\Phi,(r)}_{\bullet}=\pa^{\Phi}_{r-\bullet}\pa^{\Phi,-1}_{r}$ on $[0,r]\times \{r<\tau\}$, and this formula straightforwardly implies
\begin{align}
\pa^{\Phi,(r)}_{r}\mathscr{V}^{\Phi,(r),*}_{r}=\mathscr{V}^{\Phi}_{r}\pa^{\Phi,-1}_{r}\>\text{ on $ \{r<\tau\}$}\label{for1}
\end{align}
(for example, by expanding $\mathscr{V}^{\Phi,(r)}_{r}$ into a path ordered exponential according to Proposition \ref{app1} and taking adjoints and
comparing the result with (\ref{poi})). Thus the self-adjointness of $T^{(U)}_{r}$ follows from the calculation
\begin{align*}
\left\langle T^{(U)}_{r}f,g\right\rangle_m&= \sum_{x\in U} \mathbb{E}^x\left[1_{\{r<\tau_U\}}\left(f(\mathbb{X}_r),\pa^{\Phi}_{r}\mathscr{V}^{\Phi,*}_{r}g(x)\right)_{\mathbb{X}_r}\right]m(x)\\
&=\sum_{x\in U} \mathbb{E}^x\left[1_{\{r<\tau_U\}}\left(f(x),\pa^{\Phi,(r)}_{r}\mathscr{V}^{\Phi,(r),*}_{r}g(\mathbb{X}_r)\right)_{x}\right]m(x)\\
&=\sum_{x\in U} \mathbb{E}^x\left[1_{\{r<\tau_U\}}\left(f(x),\mathscr{V}^{\Phi}_{r}\pa^{\Phi,-1}_{r}g(\mathbb{X}_r)\right)_{x}\right]m(x)\\
&=\left\langle f,T^{(U)}_{r}g\right\rangle_m,
\end{align*}
where we have used the reversibility of $\IMM$ for the second equality, and (\ref{for1}) for the third equality.

Finally, the strong continuity follows from the semigroup property and the strong continuity at $t=0$, which is a simple consequence of the path properties of $\mathbb{X}$.

\vspace{1.2mm}

\emph{Proof of claim} (ii): By the definitions of $T^{(U)}_{t}$ and $\mathbb{X}_{t}$ we can write
\begin{align}\label{E:conv-2}
\frac{1}{t}\left(T^{(U)}_{t}f(x)-f(x)\right)&=\frac{1}{t}\left(\mathbb{E}^x\left[1_{\{N(t)=0\}}\mathscr{V}^{\Phi}_{t}//_{t}^{\Phi,-1}f(\mathbb{X}_{t})\right]-f(x)
\right)\nonumber\\
&+\omega_{t}(x)+\rho_t(x),
\end{align}
where
\[
\omega_{t}(x):=\frac{1}{t}\mathbb{E}^x\left[1_{\{N(t)=1,\mathbb{X}_{\tau_1}\in U\}}
\mathscr{V}_{t}^{\Phi}//_{t}^{\Phi,-1}f(\mathbb{X}_{t})\right],
\]
\[
\rho_t(x):=\frac{1}{t}\mathbb{E}^x\left[1_{\{2\leq N(t)<\infty, \mathbb{X}_{\tau_{N(t)}}\in U\}}\mathscr{V}_{t}^{\Phi}//_{t}^{\Phi,-1}f(\mathbb{X}_{t})\right].
\]
Using Proposition \ref{app2} (i), the term $\rho_t(x)$ in~(\ref{E:conv-2}) may be estimated as
\begin{align}\label{E:conv-3-2}
|\rho_t(x)|\leq \mathrm{e}^{t\max_{x\in U}|V(x)|_{x}}\cdot \frac{1}{t}\mathbb{E}^x\left[1_{\{2\leq N(t)<\infty\}}|f(\mathbb{X}_{t})|\right],
\end{align}
which converges to $0$ as $t\to 0+$ by~\cite[Lemma 4.5]{GKS-13}. \\
We will now rewrite the first summand on the right hand side of~(\ref{E:conv-2}). By the definition of $\pa^{\Phi}$, we have $\pa^{\Phi}_s=\mathbf{1}_{x}$ for all $0\leq s\leq t$, $\mathbb{P}^x$-a.s. in $\{N(t)=0\}$, so that Proposition \ref{poe} (i) implies $\mathscr{V}_{t}^{\Phi}=\mathrm{e}^{-tV(x)}$, $\mathbb{P}^x$-a.s. in $\{N(t)=0\}$, and we get

\begin{align}\label{E:conv-3}
&\frac{1}{t}\left(\mathbb{E}^x\left[1_{\{N(t)=0\}}\mathscr{V}_{t}^{\Phi}//_{t}^{\Phi,-1}f(\mathbb{X}_{t})\right]-f(x)\right)\nn\\
&=
\frac{1}{t}\Big(\mathbb{P}^x(N(t)=0)\mathrm{e}^{-tV(x)}f(x)-f(x)\Big).
\end{align}
Additionally, note that
\begin{align}\label{E:conv-4}
\mathbb{P}^x(N(t)=0)=\mathbb{P}^x(t<\tau_{1})=\mathbb{P}^x(\textrm{deg}_{m,b}(x)t<\xi_1)=\mathrm{e}^{-\textrm{deg}_{m,b}(x)t},
\end{align}
where the last equality follows since  $\xi_1$ is an exponential random variable with parameter $1$.
Combining~(\ref{E:conv-3}) and~(\ref{E:conv-4}), we see that the first summand on the right hand side of~(\ref{E:conv-2}) converges to
$$
-\Big(V(x)f(x)+\textrm{deg}_{m,b}(x)f(x)\Big),\>\>\>\text{ as $t\to 0+$.}
$$
We now consider the summand $\omega_{t}(x)$ in~(\ref{E:conv-2}). For $y\in U$, by the definition of $//^{\Phi}$, we have $//_{t}^{\Phi,-1}=\Phi_{y,x}$, $\mathbb{P}^x$-a.s. in $\{N(t)=1,\mathbb{X}_{\tau_1}=y\}$. Hence,
\begin{align*}\label{E:conv-5}
\omega_{t}(x)=\frac{1}{t}\sum_{y\in U}\mathbb{E}^x\left[1_{\{N(t)=1,\mathbb{X}_{\tau_1}=y\}}
\mathscr{V}_{t}^{\Phi}\Phi_{y,x}f(y)\right]=I^{(1)}_{t}(x)+I^{(2)}_{t}(x),
\end{align*}
where
\[
I^{(1)}_{t}(x):=\frac{1}{t}\sum_{y\in U}\mathbb{E}^x\left[1_{\{N(t)=1,\mathbb{X}_{\tau_1}=y\}}
(\mathscr{V}_{t}^{\Phi}-\mathbf{1}_{x})\Phi_{y,x}f(y)\right],
\]
and
\[
I^{(2)}_{t}(x):=\frac{1}{t}\sum_{y\in U}\mathbb{E}^x\left[1_{\{N(t)=1,\mathbb{X}_{\tau_1}=y\}}\Phi_{y,x}f(y)\right].
\]
Note that we can write
\begin{align}\label{E:conv-5-1}
I^{(2)}_{t}(x)&=\sum_{y\in U}\frac{1}{t}\mathbb{P}^x(N(t)=1,\mathbb{X}_{\tau_1}=y)\Phi_{y,x}f(y).
\end{align}
By the proof of~\cite[Lemma 4.5]{GKS-13} it follows that
\begin{align}\label{E:conv-5-2-2-2}
\frac{1}{t}\mathbb{P}^x(N(t)=1,\mathbb{X}_{\tau_1}=y)\to \frac{b(x,y)}{m(x)},\qquad\textrm{as }t\to 0+,
\end{align}
which together with~(\ref{E:conv-5-1}) gives
\[
I^{(2)}_{t}(x)\to \frac{1}{m(x)}\sum_{y\in U}b(x,y)\Phi_{y,x}f(y),\qquad\textrm{as }t\to 0+.
\]
Next, using the pathwise continuity of $\mathscr{V}^{\Phi}$ together with~(\ref{E:conv-5-2-2-2}) we have
\begin{align*}
|I^{(1)}_{t}(x)|_x&\leq \sum_{y\in U}\frac{1}{t}\mathbb{P}^x(N(t)=1,\mathbb{X}_{\tau_1}=y)\left|\mathscr{V}_{t}^{\Phi}-\mathbf{1}_{x}\right|_{x}|\Phi_{y,x}f(y)|_x\\\nonumber
&\to 0,\qquad\textrm{as }t\to 0+;
\end{align*}
hence,
\begin{align*}
\omega_{t}(x)\to \frac{1}{m(x)}\sum_{y\in U}b(x,y)\Phi_{y,x}f(y),\qquad\textrm{as }t\to 0+.
\end{align*}
Going back to~(\ref{E:conv-2}) and taking the limit as $t\to 0+$, we arrive at
\begin{align*}
&\lim_{t\to0+}\frac{T^{(U)}_{t}f(x)-f(x)}{t}\nonumber\\
&=-\Big(V(x)f(x)+\textrm{deg}_{m,b}(x)f(x)\Big)+\frac{1}{m(x)}\sum_{\overset{y\in U}{y\sm x}}b(x,y)\Phi_{y,x}f(y)\\
&=-H^{(U)}_{\Phi,V}f(x),
\end{align*}
which is the claim.

\subsection{Proof of Theorem \ref{FK} in the general case}

The proof will be divided into three parts:\vspace{1.2mm}

Claim (i): \emph{Formula (\ref{fds}) holds, if $V\geq -C$ for some number $C>0$.} \vspace{1.2mm}

\emph{Proof}: By adding a constant to $V$ if necessary, we can assume $V\geq 0$.

In this case, $H_{\Phi,V}$ is well-defined as the self-adjoint semibounded operator corresponding to a densely defined sum of nonnegative closed quadratic forms.

Let $(X_n)_{n\in \mathbb{N}_{0}}$ be an exhausting sequence for $X$, let $Q^{(X_n)}_{\Phi,V}$ and $H^{(X_n)}_{\Phi,V}$ be as in Section~\ref{SS:pf-fki-finite} with $U=X_n$, let $i_{X_n}$ be the canonic inclusion operator
\[
i_{X_n}\colon \Gamma_{\ell_{m}^2}(X_n,F)\hookrightarrow  \Gamma_{\ell_{m}^2}(X,F),
\]
defined by extending sections to zero outside of the set $X_n$, and let $\pi_{X_n}:=i^{*}_{X_n}$ be the adjoint of $i_{X_n}$.
Now using~\cite[Theorem C.2]{GKS-13} with $q_n:=Q^{(X_n)}_{\Phi,V}$, $q:=Q_{\Phi,V}$ (where we rely on $V\geq 0$ for a monotone convergence argument for integrals), we get
\[
\lim_{n\to\infty}i_{X_n}\mathrm{e}^{-tH^{(X_n)}_{\Phi,V}}\pi_{X_n}f(x)=\mathrm{e}^{-tH_{\Phi,V}}f(x);
\]
hence, in view of (\ref{local}), it remains to show that
\begin{align}\label{E:to-show-fki}
&\lim_{n\to\infty}\mathbb{E}^x\left[1_{\{t<\tau_{X_n}\}}\mathscr{V}_{t}^{\Phi}//_{t}^{\Phi,-1}\pi_{X_n}f(\mathbb{X}_{t})\right]\nonumber\\
&=\mathbb{E}^x\left[1_{\{t<\tau\}}\mathscr{V}_{t}^{\Phi}//_{t}^{\Phi,-1}f(\mathbb{X}_{t})\right].
\end{align}
Note that the sequence $\tau_{X_n}$ converges to $\tau$ in a monotone increasing way. Additionally, using $V\geq 0$ we get
\begin{equation*}\label{E:abs-fact-1}
1_{\{t<\tau_{X_n}\}}|\mathscr{V}_{t}^{\Phi}|_{x}\leq 1_{\{t<\tau\}}\>\>\text{ $\mathbb{P}^x$-a.s.}
\end{equation*}
from Proposition \ref{app2} (i); hence
\[
\left|1_{\{t<\tau_{X_n}\}}\mathscr{V}_{t}^{\Phi}//_{t}^{\Phi,-1}\pi_{X_n}f(\mathbb{X}_{t})\right|_x\leq
1_{\{t<\tau\}}|f|(\mathbb{X}_{t})\>\>\text{ $\mathbb{P}^x$-a.s..}
\]
Finally, as we have
\[
\mathbb{E}^x\left[1_{\{t<\tau\}}|f|(\mathbb{X}_{t})\right]=\mathrm{e}^{-tH}|f|(x)<\infty,
\]
(\ref{E:to-show-fki}) follows from dominated convergence.

\vspace{1.2mm}

Claim (ii): \emph{If $V$ admits a decomposition $V=V^{+}-V^{-}$ into potentials $V^{\pm}\geq 0$ such that $Q_{|V^-|}$ is $Q$-bounded with bound $<1$, then $Q_{V^{-}}$ is $Q_{\Phi,0}$-bounded with bound $<1$.} \vspace{1.2mm}

\emph{Proof}: By claim (i) we have
$$
\left|\mathrm{e}^{-tH_{\Phi,0}}f(x)\right|_x\leq \mathrm{e}^{-tH}|f|(x),
$$
which implies the assertion by an abstract functional analytic fact (cf. Theorem D.6 (b) in \cite{G2}).\vspace{1.2mm}

Claim (iii): \emph{Formula (\ref{fds}) holds for general $V$.} \vspace{1.2mm}

\emph{Proof}: As this case can be carried out precisely with the arguments from pp. 4648 - 4649 in \cite{G2}, we only sketch the proof:

Using the spectral calculus on the fibers of $F\to X$, we can define a sequence of potentials by $V_n(x):=\max(-n ,V(x))$. Then we have $-n\leq V_n$ for all $n$, and $V_n\searrow V$ as $n\to\infty$ pointwise in the sense of self-adjoint operators, and using convergence of monotonely decreasing quadratic forms we get
\[
\lim_{n\to\infty}\mathrm{e}^{-tH_{\Phi,V_n}}f(x)=\mathrm{e}^{-tH_{\Phi,V}}f(x).
\]
Furthermore, by claim (i) we have
\begin{align}
\mathrm{e}^{-tH_{\Phi,V_n}}f(x)=\mathbb{E}^x\left[1_{\{t<\tau\}} \mathscr{V}_{n,t}^{\Phi}//_{t}^{\Phi,-1}f(\mathbb{X}_{t})\right],\nn
\end{align}
so that it remains to prove
\begin{align*}
\lim_{n\to\infty}\mathbb{E}^x\left[1_{\{t<\tau\}} \mathscr{V}_{n,t}^{\Phi}//_{t}^{\Phi,-1}f(\mathbb{X}_{t})\right]=\mathbb{E}^x\left[1_{\{t<\tau\}} \mathscr{V}_{t}^{\Phi}//_{t}^{\Phi,-1}f(\mathbb{X}_{t})\right],
\end{align*}
where $\mathscr{V}_{n,t}^{\Phi}$ is defined by~(\ref{int}) with $V$ replaced by $V_n$.

The latter follows from dominated convergence: Indeed, Proposition \ref{app2} (ii) implies
\begin{align*}
&1_{\{t<\tau\}}\left|\mathscr{V}_{n,t}^{\Phi}//_{t}^{\Phi,-1}f(\mathbb{X}_{t})- \mathscr{V}_{t}^{\Phi}//_{t}^{\Phi,-1}f(\mathbb{X}_{t})\right|_x\nn\\
&\leq 1_{\{t<\tau\}}\left|\mathscr{V}_{n,t}^{\Phi}- \mathscr{V}_{t}^{\Phi}\right|_x|f(\mathbb{X}_{t})|_{\mathbb{X}_{t}}
\to 0\>\>\text{ $\mathbb{P}^x$-a.s.,}
\end{align*}
and defining $w:X\to\IR$ by (\ref{sxs}), we have
\[
\left|1_{\{t<\tau\}}\mathscr{V}_{n,t}^{\Phi}//_{t}^{\Phi,-1}f(\mathbb{X}_{t})\right|_x\leq
1_{\{t<\tau\}}\mathrm{e}^{-\int^t_0w(\mathbb{X}_{s})\Id s}|f|(\mathbb{X}_{t})\>\>\text{ $\mathbb{P}^x$-a.s..}
\]
Finally, the latter random variable satisfies
\[
\mathbb{E}^x\left[1_{\{t<\tau\}}\mathrm{e}^{-\int^t_0w(\mathbb{X}_{s})\Id s}|f|(\mathbb{X}_{t})\right]=\mathrm{e}^{-tH_{0,w}}|f|(x)<\infty
\]
by the scalar Feynman-Kac formula (cf. Theorem 4.1 in \cite{GKS-13}), and the proof is complete.$\hfill\blacksquare$

\subsection*{Acknowledgements}
The first author (B.G.) would like to thank Matthias Keller for many helpful discussions on graphs. B.G has been financially supported by the
SFB~647 ``Space---Time---Matter''.

\appendix

\section{Appendix}

For the convenience of the reader, we record some facts on linear initial value problems here. Let $\IHH$ be a real or complex, finite dimensional Hilbert space, where as usual $\|\bullet\|$ will denote both, the underlying norm and operator norm. We fix some $\tau\in (0,\infty]$ for the following results.

\begin{Proposition}\label{app1} Let $A\in \mathsf{L}^1_{\mathrm{loc}}([0,\tau),\ILL(\IHH))$ and let $0\leq t< \tau$.

\emph{(i)} There is a unique solution $\mathscr{A}:[0,\tau)\to\ILL(\IHH)$ of the initial value problem
$$
\mathscr{A}(t)=\mathbf{1}+\int^t_0 \mathscr{A}(s)A(s) \Id s.
$$
In fact, one has
\begin{align}
\mathscr{A}(t)=\mathbf{1}+\sum_{n=1}^{\infty}\int_{t\sigma_n}A(s_1)\cdots A(s_n) \,\Id s_1\,\cdots\,\Id s_n ,
\end{align}
where
\[
t\sigma_n:=\{(s_1,s_2,\dots,s_n)| \>0\leq s_1\leq s_2\leq\cdots\leq s_n\leq t\},
\]
and where the series converges absolutely and locally uniformly in $t$. More specifically, for all $n\in \IN$ one has
\begin{align}
&\int_{t\sigma_n}\left\|A(s_1)\cdots A(s_n)\right\| \Id s_1\cdots\Id s_n \leq \f{\left(\int^t_0 \|A(s)\|\Id s\right)^n}{n!}.\nn
\end{align}

\emph{(ii)} $\mathscr{A}^*$ is the unique solution of the initial value problem
$$
\mathscr{A}^*(t)=\mathbf{1}+\int^t_0 A^*(s)\mathscr{A}^*(s) \Id s.
$$

\emph{(iii)} $\mathscr{A}$ is invertible and $\mathscr{A}^{-1}$ is the unique solution of the initial value problem
$$
\mathscr{A}^{-1}(t)=\mathbf{1}-\int^t_0 A(s)\mathscr{A}^{-1}(s) \Id s.
$$
\end{Proposition}

\begin{Proposition}\label{app2} Let $A\in \mathsf{L}^1_{\mathrm{loc}}([0,\tau),\ILL(\IHH))$ and let $0\leq t< \tau$.

\emph{(i)} Assume that $A(\bullet)$ is self-adjoint and that there is a real-valued function $C\in \mathsf{L}^1[0,t]$ such that for all $h\in \IHH$ one has
$$
\left\langle A(\bullet)h,h\right\rangle \leq C(\bullet)\|h\|^2 \>\text{ a.e. in $[0,t]$.}
$$
Then it holds that
$$
\left\|\mathscr{A}(t)\right\|\leq\mathrm{e}^{\int^t_0 C(s)\Id s}.
$$

\emph{(ii)} Let $\tilde{A}\in \mathsf{L}^1_{\mathrm{loc}}([0,\tau),\ILL(\IHH))$. Then (with an obvious notation) one has
\begin{align*}
&\left\|\mathscr{A}(t)-\mathscr{\tilde{A}}(t)\right\|\leq \mathrm{e}^{2\int^t_0 \|A(s)\|\Id s+\int^t_0 \|\tilde{A}(s)\|\Id s}\int^t_0 \|A(s)-\tilde{A}(s)\|\Id s.
\end{align*}
\end{Proposition}

We refer the reader to the appendix C of \cite{G2} and the references therein for the proofs of these statements.

\end{document}